\documentclass{amsart}
\usepackage{amssymb}
\usepackage{latexsym}
\usepackage{amsmath}
\usepackage{euscript}
\usepackage{graphics}
\usepackage{todonotes}

\def\hypergeom#1#2#3#4#5{{}_#1 F_{#2}\left({#3\atop#4}; \ #5\right)}

      \def\dR{{\mathbb R}}

   \def\dZ{{\mathbb Z}}

\def\bm\chi{\mbox{\boldmath$\chi$}}

\let\xker=\ker \def\ker{{\xker\,}}

\newtheorem{theorem}{Theorem}[section]
\newtheorem{proposition}[theorem]{Proposition}
\newtheorem{corollary}[theorem]{Corollary}

\theoremstyle{remark}

\newtheorem{remark}[theorem]{Remark}

\numberwithin{equation}{section}

\date{\today}

\author[M.~Derevyagin]{Maxim~Derevyagin}
\address{
MD,
Department of Mathematics\\
University of Connecticut\\
341 Mansfield Road, U-1009\\
Storrs, CT 06269-1009, USA}
\email{derevyagin.m@gmail.com}

\author[J.~Geronimo]{Jeffrey~S.~Geronimo}
\address{JG, School of Mathematics, Georgia Institute of Technology,
686 Cherry Street,
Atlanta, GA 30332--0160, USA}
\email{jeffrey.geronimo@math.gatech.edu}

\dedicatory{Dedicated to the memory of Richard Askey}


\begin{document}

\title[Connection coefficients generalized bispectrality]{Connection coefficients for ultraspherical polynomials with argument doubling and generalized bispectrality}

\begin{abstract}
We start by presenting a generalization of a discrete wave equation that is satisfied by the entries of the matrix coefficients of the refinement equation corresponding to the multiresolution analysis of Alpert. The entries are in fact functions of two discrete variables and they can be expressed in terms of the Legendre polynomials. Next, we generalize these functions to the case of the ultraspherical polynomials and show that these new functions obey two generalized eigenvalue problems in each of the two discrete variables, which constitute a generalized bispectral problem. At the end, we make some connections to other problems.
\end{abstract}

\maketitle

\section{Introduction}

Let $\{P_n\}_{n=0}^{\infty}$ and $\{Q_n\}_{n=0}^{\infty}$  be two families of orthonormal polynomials whose orthogonality measures are $d\mu$ and $d\nu$, respectively. Then one can see that
\[
P_i(t)=\sum_{j=0}^{i}c_{i,j}Q_j(t),
\]
where the coefficients $c_{i,j}$ can be found in the following way
\[
c_{i,j}=\int P_i(t)Q_j(t)\,d\nu(t).
\]
These coefficients are called connection coefficients and their nonnegativity for some particular cases of the ultraspherical polynomials is useful in the proof of the positivity of a certain $_3F_2$ function, which in turn, based on the work of Gasper and Askey and Gasper, played a significant role in the first proof of the Bieberbach conjecture \cite{AAR}. Also there has been much work proving the nonnegativity of integrals of products of orthogonal polynomials times certain functions which was initiated by Askey in the late 1960's. These studies have been  stimulated by the fact that some of those integrals have combinatorial interpretations (see \cite{IKZ13}). 

Another instance that we would like to mention is that in some early work leading to the theory of bispectral problems a matrix $S_1$, whose entries are
\[
(S_1)_{i,j}=\int_a^{\Omega} P_i(t)P_j(t)\,d\mu(t)
\]
for some real $a$ and $\Omega$, was considered (for instance, see \cite{Grunbaum83}). The question was to find eigenvectors of $S_1$. Since  
$S_1$ is a full matrix, this is not an easy task. However, it was proposed to find a tridiagonal matrix commuting with $S_1$ in order to reduce the original problem to a problem of finding eigenvectors of the tridiagonal matrix, which is an easier and well understood problem. It was shown to be possible to construct such tridiagonal matrices for some families of orthogonal polynomials and this is one of the fundamental ideas in the theory of bispectral problems.  

The last instance to bring up here is that in \cite{GM15} the Alpert multiresolution analysis was studied in detail and important in this study was the integral
$$
f_{i,j}=\int_0^1 \hat p_i(t)\hat p_j(2t-1)dt,
$$
where $\hat p_i$ is the orthonormal Legendre polynomial, i.e. $\hat p_j(t)=k_j t^j +\text{lower degree terms}$ with $k_j>0$ and for any two nonnegative integers $k$ and $l$ we have 
\[
\int_{-1}^1 \hat p_k(t)\hat p_l(t)dt=\begin{cases}
                                     0, \,\,k\ne l;\\
                                     1,\,\, k=l.                                  
                                  \end{cases}
\]
These coefficients are entries in the refinement equation associated with this multiresolution analysis. The fact that the Legendre polynomials are involved in the above integral allowed the authors in \cite{GM15} to obtain many types of recurrence formulas in $i$ and $j$ including a generalized eigenvalue problem in each of the indices. These two equations together give rise to  a bispectral generalized eigenvalue problem.

We begin by discussing a common property of the coefficients in all the above-mentioned cases: they satisfy a generalized 2D discrete wave equation. Then we observe numerically that a damped oscillatory behavior takes place in the case of the ultraspherical generalization of the  coefficients $f_{i,j}$. In particular with
$$
f^{(\lambda)}_{i,j}=\int_0^1 \hat p^{(\lambda)}_i(t)\hat p^{(\lambda)}_j(2t-1)(t(1-t))^{\lambda-1/2}dt,
$$ 
where ${\hat p^{(\lambda)}_i}$ are the orthonormal ultraspherical polynomials and $\lambda>\ -1/2$ we find  the asymptotic formula 
\begin{equation*}
f^{(\lambda)}_{i,j}=k_j\frac{\cos\left(\pi\left(j + \frac{\lambda}{2} - \frac{i}{2} + \frac{1}{4}\right)\right)}{\sqrt{\pi}i^{\lambda + 1/2}}+O\left(\frac{1}{i^{\lambda+3/2}}\right),
\end{equation*}
where
\begin{equation*}
k_j=\frac{1}{2^{j+1-2\lambda}}\sqrt{\frac{(2\lambda)_j}{j!(\lambda)_j(\lambda+1)_j\lambda\Gamma(2\lambda)}}\Gamma(2j+2\lambda+1)(\lambda+\frac{1}{2})_j 
\end{equation*}
which confirms the damped oscillatory behavior. We also derive some related properties and show that $f^{(\lambda)}_{i,j}$ satisfy a bispectral generalized eigenvalue problem of the form
\[
\begin{split}
\tilde A_if^{(\lambda)}_{i,j}=(j+\lambda-\frac{1}{2})(j+\lambda+\frac{1}{2})B_if^{(\lambda)}_{i,j},\\
\hat A_jf^{(\lambda)}_{i,j}=(i+\lambda+\frac{1}{2})(i+\lambda-\frac{1}{2})\hat B_jf^{(\lambda)}_{i,j},
\end{split}
\] 
where $\tilde A_i$, $B_i$ are tridiagonal operators or second order linear difference operators acting on $i$ and $\hat A_j$, $\hat B_j$ are tridiagonal operators acting on $j$. Each of the two above-given relations is a generalized eigenvalue problem and the theory of such problems is intimately related to biorthogonal rational functions (for instance, see \cite{GM98}, \cite{IM95}, \cite{SZh00}). 

The paper is organized as follows. In Section 2 a vast generalization of the above integral is shown to give rise to a 2D wave equation and solutions to the special case of the above integral are plotted to show the oscillations. In Section 3 the Legendre case above is analyzed and various properties of the coefficients $f_{i,j}$ are  derived. One point of this section is to derive the orthogonality property of these coefficients using that they come from special functions. In Section 4 the Legendre polynomials are replaced by the ultraspherical polynomials and their scaled weight. 
%While this allows the introduction of the parameter $\lambda$ the connection to multiresolution analysis is lost due to the weight. 
Here it is shown that the coefficients $f^{(\lambda)}_{i,j}$ satisfy a wave equation and also a bispectral generalized eigenvalue problem. Two proofs are given developing the generalized eigenvalue problem. One is based  on the fact that the polynomials satisfy a  differential equation and has the flavor of the proof given in \cite{Grunbaum83} and the second follows from the formula for $f_{i,j}^{(\lambda)}$ in terms of a ${}_2 F_1$ hypergeometric function. The two proofs emphasize different aspects of the problem that maybe useful when viewing other orthogonal polynomial systems. In Section 5 connections are made to various other problems.

\section{The 2D discrete wave equation}

Let $\{P_n\}_{n=0}^{\infty}$ and $\{Q_n\}_{n=0}^{\infty}$ be two families of orthonormal polynomials with respect to two probability measures or, equivalently, two families that obey the three-term recurrence relations
\[
a_{n+1}P_{n+1}(t)+b_nP_n(t)+a_nP_{n-1}(t)=tP_n(t), \quad n=0,1,2,\dots
\] 
and
\[
c_{n+1}Q_{n+1}(t)+d_nQ_n(t)+c_nQ_{n-1}(t)=tQ_n(t), \quad n=0,1,2,\dots,
\] 
where the coefficients $a_n$ and $c_n$ are positive and the coefficients $b_n$ and $d_n$ are real. In particular, the first relations are
\[
a_{1}P_{1}(t)+b_0P_0(t)=tP_0(t), \quad c_{1}Q_{1}(t)+d_0Q_0(t)=tQ_0(t).
\]
Therefore we can set $a_0=c_0=0$ for the coefficients to be defined for $n=0,1,2,\dots$. Since the families are orthonormal with respect to probability measures we know that 
\[
P_0=1, \quad Q_0=1,
\] 
which are the initial conditions that allow to reconstruct  each of the systems from the corresponding recurrence relation. It should be stressed here that by imposing these particular initial conditions we implicitly assume that the corresponding orthogonality measures are probability measures.
 
In addition, suppose we are given a measure $\sigma$ on $\dR$ with finite moments.
Then, let us consider the coefficients 
\begin{equation}\label{Cijg}
u_{i,j}=\int_{\dR} P_i(t)Q_j(\alpha t+\beta) d\sigma(t),
\end{equation}
where $\alpha\ne 0$ and $\beta$ are complex numbers. It turns out that these coefficients constitute a solution of a generalized wave equation on the two dimensional lattice.
%(see also \cite{BK87} where a similar generalization was considered).

\begin{theorem}[cf. Theorem 2.1 from \cite{IKZ13}]\label{GdWaveTH} We have that
\begin{equation}\label{dWaveEqG}
%\begin{split}
a_{i+1}u_{i+1,j}+b_iu_{i,j}+a_iu_{i-1,j}
=\frac{c_{j+1}}{\alpha}u_{i,j+1}+\frac{d_{j}-\beta}{\alpha}u_{i,j}+\frac{c_j}{\alpha}u_{i,j-1}
%\end{split}
\end{equation}
for $i,j=0$, $1$, $2$, \dots.
\end{theorem}
\begin{proof}
%Let us introduce the following partial difference operators on the lattice 
%\[
%\partial^2_{ii}u_{i,j}:=a_{i+1}u_{i+1,j}+b_iu_{i,j}+a_iu_{i-1,j}
%\]
%and
%\[
%\Delta^2_{jj}u_{i,j}:=c_{j+1}u_{i,j+1}+d_ju_{i,j}+c_ju_{i,j-1}.
%\]
From \eqref{Cijg} and the three-term recurrence relations we get that 
\begin{align*}
&a_{i+1}u_{i+1,j}+b_iu_{i,j}+a_iu_{i-1,j}\\&=
\int_{\dR}(a_{i+1}P_{i+1}(t)+b_iP_i(t)+a_iP_{i-1}(t))Q_j(\alpha t+\beta) d\sigma(t)&
\\&=
\int_{\dR} tP_i(t)Q_j(\alpha t+\beta) d\sigma(t)\\&
=\frac{1}{\alpha}\int_{\dR}P_i(t)(\alpha t+\beta)Q_j(\alpha t+\beta) d\sigma(t)-\frac{\beta}{\alpha}\int_{\dR}P_i(t)Q_j(\alpha t+\beta) d\sigma(t)\\&=
\frac{1}{\alpha}
\int_{\dR}P_i(t)(c_{j+1}Q_{j+1}(\alpha t+\beta)+d_jQ_j(\alpha t+\beta)+c_jQ_{j-1}(\alpha t+\beta))d\sigma(t)\\&-\frac{\beta}{\alpha}u_{i,j}\\&
=\frac{c_{j+1}}{\alpha}u_{i,j+1}+\frac{d_{j}}{\alpha}u_{i,j}+\frac{c_j}{\alpha}u_{i,j-1}-\frac{\beta}{\alpha}u_{i,j}
\end{align*}
and thus \eqref{dWaveEqG} holds.
%Therefore, we have arrived at the discrete wave equation on the lattice
%\begin{equation}\label{dWaveEqH}
%\partial^2_{ii}u_{i,j}=\frac{1}{\alpha}\partial^2_{jj}u_{i,j}-\frac{\beta}{\alpha}u_{i,j},
%\end{equation}
%which can be easily transformed into \eqref{dWaveEqG}.
\end{proof}

\begin{remark} Given an equation of the form \eqref{dWaveEqG} then due to the Favard theorem the coefficients will uniquely determine the families $\{P_n\}_{n=0}^{\infty}$ and $\{Q_n\}_{n=0}^{\infty}$ of orthonormal polynomials. The measure $\sigma$ is responsible for the initial state when $j=0$ and $j$ can be thought of as a discrete time. Namely, for a solution of the form \eqref{Cijg} to exist they need to satisfy the initial condition
\[
u_{i,0}=\int_{\dR} P_i(t)d\sigma(t),
\]
which means that given initial function $u_{i,0}$ of the discrete space variable $i$, $\sigma$ needs to be found. The latter problem is a generalized moment problem and in this particular case it is equivalent to a Hamburger moment problem. 

It is also worth mentioning here that another type of cross-difference equations on $\dZ^2_+$ was recently discussed in \cite{ADvA} and the construction was based on multiple orthogonal polynomials. Type I Legendre-Angelesco multiple orthogonal polynomials also arise in the wavelet construction proposed by Alpert \cite{GIVA}.  
\end{remark}

Next, consider a particular case of the above scheme where $P_n$ and $Q_n$ are both  orthonormal Legendre polynomials $\hat p_{n}$ and so  verify the three-term recurrence relation
\[
\frac{(n+1)}{\sqrt{(2n+1)(2n+3)}}\hat p_{n+1}(t)+\frac{n}{\sqrt{(2n-1)(2n+1)}}\hat p_{n-1}(t)=t\hat p_n(t),
\] 
for $n=0, 1, 2, \dots$. Set $\sigma$ to be the Lebesgue measure on the interval $[0,1]$. As a result, the coefficients \eqref{Cijg} take the form   
\begin{equation}\label{Cij}
f_{i,j}=\int_{0}^1 \hat p_i(t)\hat p_j(2t-1) dt.
\end{equation}
It is not so hard to see that the polynomials $\hat p_j(2t-1)$ are orthogonal on the interval $[0,1]$ with respect to the Lebesgue measure, consequently
\begin{equation}\label{HalfZ}
f_{i,j}=0, \quad j>i=0, 1, 2, \dots.
\end{equation}

Since the coefficients of the three-term recurrence relation for the Legendre polynomials are explicitly known, the coefficients of equation \eqref{dWaveEqG} become explicit as well. The following Corollary can be found in \cite{GM15}.

\begin{corollary} The function $f_{i,j}$ satisfies,
\begin{equation}\label{dWaveEq}
\begin{split}
\frac{j+1}{\sqrt{(2j+1)(2j+3)}}f_{i,j+1}+f_{i,j}+\frac{j}{\sqrt{(2j-1)(2j+1)}}f_{i,j-1}=\\
=\frac{2(i+1)}{\sqrt{(2i+1)(2i+3)}}f_{i+1,j}+\frac{2i}{\sqrt{(2i-1)(2i+1)}}f_{i-1,j}
\end{split}
\end{equation}
for $i,j=0$, $1$, $2$, \dots.
\end{corollary}

Below is the MATLAB generated graphical representation of some behavior of the solution $f_{i,j}$ to equation \eqref{dWaveEq}, which is a generalization of the discretized wave equation.

\begin{figure}[h!]
\includegraphics[width=\linewidth]{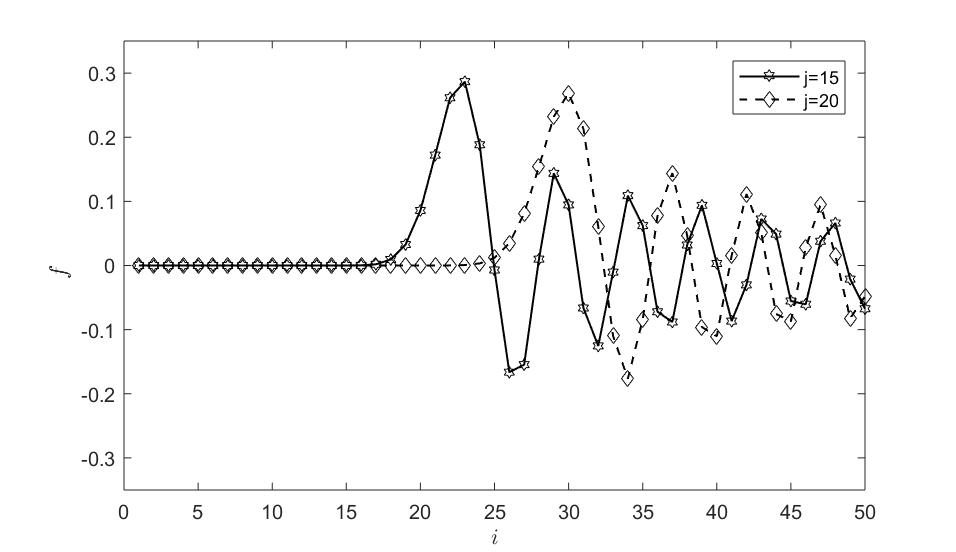}
\caption{This picture demonstrates the moving wave. Here, one can see two graphs of the function $f=f(i)=f_{i,j}$ of the discrete space variable $i$ at the two different discrete times $j=15$ and $j=20$.}
\label{Fig1}
\end{figure}

To sum up, we would like to point out here that the form \eqref{Cijg} of solutions of the discrete wave equations is very useful for understanding the behavior of solutions because there are many asymptotic results for a variety of families of orthogonal polynomials.

\section{Some further analysis of the coefficients $f_{i,j}$}

In this section, we will obtain some properties of the coefficients $f_{i,j}$ based on the intuition and observations developed in \cite{GM15}. In particular, we will rederive and expand upon some orthogonality properties of the coefficients $f_{i,j}$. 

We begin with the following statement, which is based on formula \eqref{Cij} and some known properties of the Legendre polynomials.

\begin{theorem}
Let $k$ and $l$ be two nonnegative integer numbers. Then one has
\begin{equation}\label{CijOrth}
\sum_{j=0}^{\infty}f_{k,j}f_{l,j}=\begin{cases}
                                     0, \,\,\text{if $k$ and $l$ are of the same parity but not equal};\\
                                     1,\,\,\text{if} \,\, k=l;\\
                                     (-1)^{\frac{k+l+1}{2}}\frac{k!l!\sqrt{2k+1}\sqrt{2l+1}}{2^{k+l-1}(k-l)(k+l+1)((\frac{k}{2})!)^2((\frac{l-1}{2})!)^2}, \,\,\text{if $k$ and $l$ are of opossite parity}.
                                  \end{cases}
\end{equation}
\end{theorem} 

\begin{proof} 
Without loss of generality, we can assume that $k\le l$. Next observe that due to \eqref{HalfZ} the left-hand side of formula \eqref{CijOrth} is truncated to 
\begin{equation*}
\sum_{j=0}^{\infty}f_{k,j}f_{l,j}=\sum_{j=0}^{k}f_{k,j}f_{l,j},
\end{equation*}
which can be written as
\begin{equation*}
\sum_{j=0}^{k}f_{k,j}f_{l,j}=\sum_{j=0}^{k}\int_{0}^1 \hat p_k(x)\hat p_j(2x-1) dx\int_{0}^1 \hat p_{l}(y)\hat p_j(2y-1) dy.
\end{equation*}
One can rewrite the expression in the following manner
\begin{equation*}
\sum_{j=0}^{k}f_{k,j}f_{l,j}=\int_{0}^1  \hat p_{l}(y) \left(\int_{0}^1 \hat p_k(x)\sum_{j=0}^{k} \hat p_j(2x-1)\hat p_j(2y-1) dx \right) dy.
\end{equation*}
Since the Christoffel-Darboux kernel  $\displaystyle{2\sum_{j=0}^{k} \hat p_j(2x-1)\hat p_j(2y-1)}$
is a reproducing kernel, we get
\begin{equation*}
\sum_{j=0}^{k}f_{k,j}f_{l,j}=2\int_{0}^1  \hat p_{k}(y) \hat p_l(y) dy.
\end{equation*}
Next recall that one can explicitly compute the quantity 
\[
\int_{0}^1  \hat p_{k}(y) \hat p_l(y) dy
\]
for any nonnegative integers $k$ and $l$. If $k$ and $l$ have the same parity the symmetry properties of the Legendre polynomials allow the above integral to be extended to the full orthonality interval $[-1,1]$ which gives the first two parts of the Theorem. The third case of formula \eqref{CijOrth} is a consequence of \cite[p.173, Art. 91, Example 2]{B59}. 
\end{proof}  

One can also compute the inner product of vectors $f_{i,j}$ taken the other way. 
 
\begin{theorem}\label{OrthOfF}
Let $k$ and $l$ be two nonnegative integer numbers. Then one has
\begin{equation}\label{CijOrth2}
\sum_{i=0}^{\infty}f_{i,k}f_{i,l}=\begin{cases}
                                     0, \,\,k\ne l;\\
                                     1/2,\,\,\text{if} \,\, k=l.                                  
                                  \end{cases}
\end{equation}
\end{theorem} 
\begin{proof}
Let $n$ be a nonnegative integer. Then we can write 
\begin{equation*}
\sum_{i=0}^{n}f_{i,k}f_{i,l}=\sum_{i=0}^{n}\int_{0}^1 \hat p_i(x)\hat p_k(2x-1) dx\int_{0}^1 \hat p_{i}(y)\hat p_l(2y-1) dy,
\end{equation*}
which can be rewritten as follows
\begin{equation*}
\sum_{i=0}^{n}f_{i,k}f_{i,l}=\int_{-1}^1  \hat p_{k}(2x-1)\chi_{[0,1]}(x) \left(\sum_{i=0}^{n}\left(\int_{-1}^1 \hat p_l(2y-1)\chi_{[0,1]}(y) \hat p_i(y) dy \right)\hat p_i(x) \right)dx.
\end{equation*}
Since the polynomials $\hat p_i$ form an orthonormal basis in $L_2([-1,1],dt)$ we know that 
\[
\sum_{i=0}^{n}\left(\int_{-1}^1 \hat p_l(2y-1)\chi_{[0,1]}(y) \hat p_i(y) dy \right)\hat p_i(x)\xrightarrow{L_2([-1,1],dt)} \hat p_{l}(2x-1)\chi_{[0,1]}(x)
\]
as $n\to\infty$. As a result we arrive at the following relation
\[
\begin{split}
\sum_{i=0}^{\infty}f_{i,k}f_{i,l}=&\int_{-1}^1 \hat p_{k}(2x-1)\chi_{[0,1]}(x)\hat p_{l}(2x-1)\chi_{[0,1]}(x)dx\\
=&\int_{0}^1 \hat p_{k}(2x-1)\hat p_{l}(2x-1)dx=\frac{1}{2}\int_{-1}^1 \hat p_{k}(t)\hat p_{l}(t)dt,
\end{split}
\]
which finally gives \eqref{CijOrth2}.
\end{proof} 

As a consequence we can say a bit more about the asymptotic behavior of the coefficients $f_{i,j}$.
\begin{corollary} Let $k$ be a fixed nonnegative integer number. Then
\[
f_{i,k}\longrightarrow 0
\]
as $i\to\infty$.
\end{corollary} 
\begin{proof}
The statement immediately follows from the fact that the series
\[
\sum_{i=0}^{\infty}f_{i,k}^2
\]
converges.
\end{proof}

\begin{remark} From \eqref{CijOrth2} one gets that
\[
\sum_{i=0}^{\infty}f_{i,j}^2=1/2
\]
for any nonnegative $j$. This means that the energy of the wave represented by $f=f(i)=f_{i,j}$ is conserved over the discrete time $j$.
\end{remark}

\begin{remark} The fact that $f_{i,k}$ can be represented as a hypergeometric function allows a more precise asymptotic estimate; see formula \eqref{asymi}. 
\end{remark}

\section{The case of ultraspherical polynomials}

In this section we will carry over our findings from the case of Legendre polynomials to the case of the family of ultraspherical polynomials which include the Legendre polynomials as a special case. 

Recall that for $\lambda>-1/2$ an ultraspherical polynomial $\hat p_n^{(\lambda)}(t)$ is a polynomial of degree $n$ that is the orthonormal polynomial with respect to the measure
\[
(1-t^2)^{\lambda-1/2}dt.
\]
In an analogous way to $f_{i,j}$, let us consider the function of the discrete variables $i$ and $j$
\begin{equation}\label{cijlam}
f^{(\lambda)}_{i,j}=\int_0^1 \hat{p}_i^{(\lambda)}(t)\hat{p}_j^{(\lambda)}(2t-1)(t(1-t))^{\lambda-1/2}dt
\end{equation}
and notice that 
\[
f_{i,j}=f_{i,j}^{(1/2)}.
\]
While this generalization allows us to consider a more general case,  the connection to multiresolution analysis seems to be lost due to the weight and there is no evident relation to multiresolution analysis for arbitrary $\lambda>-1/2$. Still, such a deformation of the coefficients $f_{i,j}$ gives an insight on how all these objects are connected to various problems some of which were mentioned in the introduction. 

Also, it is worth mentioning that the polynomials $\hat{p}_j^{(\lambda)}(2t-1)$ are orthogonal with respect to the measure
\[
(t(1-t))^{\lambda-1/2}dt.
\]
Next since the ultraspherical polynomials satisfy the three-term recurrence relation \cite{Szego}
\[
 \frac{1}{2}\sqrt{\frac{(n+1)(n+2\lambda)}{(n+\lambda)(n+\lambda+1)}}\hat{p}_{n+1}^{(\lambda)}(t)+
 \frac{1}{2}\sqrt{\frac{n(n+2\lambda-1)}{(n+\lambda-1)(n+\lambda)}}\hat{p}_{n-1}^{(\lambda)}(t)=t\hat{p}_{n}^{(\lambda)}(t)
\]
the following corollary of Theorem \ref{GdWaveTH} is immediate.

\begin{corollary} The function $f^{(\lambda)}_{i,j}$ satisfies
\begin{equation}\label{dWaveUP}
\begin{split}
\frac{1}{2}\sqrt{\frac{(j+1)(j+2\lambda)}{(j+\lambda)(j+\lambda+1)}}f_{i,j+1}^{(\lambda)}+f_{i,j}^{(\lambda)}+
\frac{1}{2}\sqrt{\frac{j(j+2\lambda-1)}{(j+\lambda-1)(j+\lambda)}}f_{i,j-1}^{(\lambda)}=\\
=\sqrt{\frac{(i+1)(i+2\lambda)}{(i+\lambda)(i+\lambda+1)}}f_{i+1,j}^{(\lambda)}+
\sqrt{\frac{i(i+2\lambda-1)}{(i+\lambda-1)(i+\lambda)}}f_{i-1,j}^{(\lambda)}
\end{split}
\end{equation}
for $i,j=0$, $1$, $2$, \dots.
\end{corollary}

As one can see from the above statement, the function $f^{(\lambda)}_{i,j}$ is a solution of a discrete wave equation and Figure \ref{Fig2} demonstrates how the function changes with $\lambda$ when $j$ is fixed.

\begin{figure}[h!]
\includegraphics[width=\linewidth]{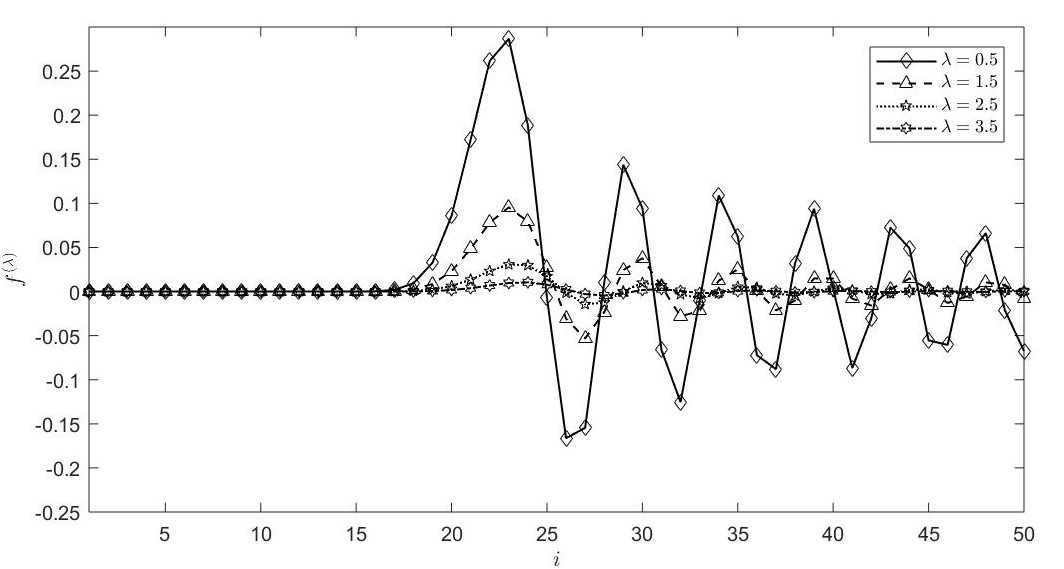}
\caption{This picture shows the $\lambda$-evolution of the function $f^{(\lambda)}=f^{(\lambda)}(i)=f_{i,j}^{(\lambda)}$ of the discrete space variable $i$ when the discrete time $j$ is fixed and $j=15$.}
\label{Fig2}
\end{figure}

It is not so hard to see that it is possible to generalize \eqref{CijOrth} and \eqref{CijOrth2}  to the case of the ultraspherical polynomials. 

\begin{theorem}
Let $k$ and $l$ be two nonnegative integer numbers. Then one has
\begin{equation}\label{CijOrthU}
\sum_{j=0}^{\infty}f_{k,j}^{(\lambda)}f_{l,j}^{(\lambda)}=\frac{1}{2^{2\lambda}}\int_{0}^1  \hat p_{k}^{(\lambda)}(y) \hat p_l^{(\lambda)}(y) (y(1-y))^{\lambda-\frac{1}{2}}dy
\end{equation}
for any $\lambda>-1/2$ and
\begin{equation}\label{CijOrth2U}
\sum_{i=0}^{\infty}f_{i,k}^{(\lambda)}f_{i,l}^{(\lambda)}=\int_{0}^1 \hat p_{k}^{(\lambda)}(2x-1)\hat p_{l}^{(\lambda)}(2x-1)x^{2\lambda-1}\left(\frac{1-x}{1+x}\right)^{\lambda-\frac{1}{2}}dx
\end{equation}
provided that $\lambda>0$.
\end{theorem} 
\begin{proof} 
As before we can assume that $k\le l$ therefore, 
\small
\begin{multline*}
\sum_{j=0}^{\infty}f_{k,j}^{(\lambda)}f_{l,j}^{(\lambda)}=\sum_{j=0}^{k}f_{k,j}^{(\lambda)}f_{l,j}^{(\lambda)}=\\
\int_{0}^1  \hat p_{l}^{(\lambda)}(y) \left(\int_{0}^1 \hat p_k^{(\lambda)}(x)\sum_{j=0}^{k} \hat p_j^{(\lambda)}(2x-1)\hat p_j^{(\lambda)}(2y-1) (x(1-x))^{\lambda-\frac{1}{2}}dx \right) (y(1-y))^{\lambda-\frac{1}{2}}dy.
\end{multline*}
\normalsize
Since the Christoffel-Darboux kernel 
\[
2^{2\lambda}\sum_{j=0}^{k} \hat p_j^{(\lambda)}(2x-1)\hat p_j^{(\lambda)}(2y-1)
\]
is a reproducing kernel in the corresponding $L_2$-space, we get
\begin{equation*}
\sum_{j=0}^{\infty}f_{k,j}^{(\lambda)}f_{l,j}^{(\lambda)}=\frac{1}{2^{2\lambda}}\int_{0}^1  \hat p_{k}^{(\lambda)}(y) \hat p_l^{(\lambda)}(y) (y(1-y))^{\lambda-\frac{1}{2}}dy.
\end{equation*}
To prove the second equality, consider the following representation of the finite sum 
\[
\sum_{i=0}^{n}f_{i,k}^{(\lambda)}f_{i,l}^{(\lambda)}=
\int_{-1}^{1}\hat p_{k}^{(\lambda)}(2x-1)\chi_{[0,1]}(x)\frac{x^{\lambda-1/2}}{(1+x)^{\lambda-1/2}}P_n(x)(1-x^2)^{\lambda-1/2}dx,
\]
where 
\[
P_n(x)=\sum_{i=0}^{n}\int_{-1}^{1}\left(\hat p_{l}^{(\lambda)}(2y-1)\chi_{[0,1]}(y))\frac{y^{\lambda-1/2}}{(1+y)^{\lambda-1/2}}p_{i}^{(\lambda)}(y)(1-y^2)^{\lambda-1/2}dy\right) \hat p_{i}^{(\lambda)}(x).
\]
If $\lambda>0$ then 
\[
P_n(x)\xrightarrow{L_2([-1,1],(1-x^2)^{\lambda-1/2}dx)}\hat p_{l}^{(\lambda)}(2x-1)\chi_{[0,1]}(x)\frac{x^{\lambda-1/2}}{(1+x)^{\lambda-1/2}}
\]
as $n\to\infty$. Next since the functional 
\[
F(g)=\int_{-1}^{1}\hat p_{k}^{(\lambda)}(2x-1)\chi_{[0,1]}(x)\frac{x^{\lambda-1/2}}{(1+x)^{\lambda-1/2}}g(x)(1-x^2)^{\lambda-1/2}dx,
\]
is continuous for $\lambda>0$ we arrive at the following
\[
\sum_{i=0}^{\infty}f_{i,k}^{(\lambda)}f_{i,l}^{(\lambda)}=\int_{0}^1 \hat p_{k}^{(\lambda)}(2x-1))\hat p_{l}^{(\lambda)}(2x-1)x^{2\lambda-1}\left(\frac{1-x}{1+x}\right)^{\lambda-\frac{1}{2}}dx
\]
which completes the proof.
\end{proof} 

\begin{remark} The first integral in the above Theorem can be evaluated with the use of the equations~(4.7.30) in \cite{Szego}. With 
$$
I^1_{k,l}=\frac{1}{2^{\lambda}}k_k k_l I^2_{k,l},
$$
where
$$
k_l=2^l\sqrt{\frac{(\lambda)_l(\lambda+1)_l}{l!(2\lambda)_l}},
$$
and
$$
I^2_{k,l}=\int_0^1 p_k^{\lambda}(y) p_i^{\lambda}(y)(y(1-y))^{\lambda-1/2}dy.
$$
With the use of the formulas alluded to above in \cite{Szego}  we find
\begin{align*}
I^2_{2k,2l}&=(-1)^{k+l}\frac{(1/2)_k(1/2)_l\Gamma(\lambda+\frac{1}{2})^2}{(k+\lambda)_k(l+\lambda)_l\Gamma(2\lambda+1)}\\&\sum_{j=0}^k{\frac{(-k)_j(k+\lambda)_j(\lambda+1/2)_{2j}}{(1)_j(1/2)_j(2\lambda+1)_{2j}}\hypergeom43{-l,l+\lambda,j+\frac{\lambda}{2}+\frac{1}{4},j+\frac{\lambda}{2}+\frac{3}{4}}{\frac{1}{2},j+\lambda+1,j+\lambda+\frac{1}{2}}{1}},
\end{align*}
\begin{align*}
I^2_{2k,2l+1}&=(-1)^{k+l}\frac{(1/2)_k(3/2)_l\Gamma(\lambda+\frac{1}{2})\Gamma(\lambda+\frac{3}{2})}{(k+\lambda)_k(l+\lambda+1)_l\Gamma(2\lambda+2)}\\&\sum_{j=0}^k{\frac{(-k)_j(k+\lambda)_j(\lambda+3/2)_{2j}}{(1)_j(1/2)_j(2\lambda+2)_{2j}}\hypergeom43{-l,l+\lambda+1,j+\frac{\lambda}{2}+\frac{3}{4},j+\frac{\lambda}{2}+\frac{5}{4}}{\frac{3}{2},j+\lambda+1,j+\lambda+\frac{3}{2}}{1}},
\end{align*}
and
\small
\begin{align*}
I^2_{2k+1,2l+1}&=(-1)^{k+l}\frac{(3/2)_k(3/2)_l\Gamma(\lambda+\frac{1}{2})\Gamma(\lambda+\frac{5}{2})}{(k+\lambda+1)_k(l+\lambda+1)_l\Gamma(2\lambda+3)}\\&\sum_{j=0}^k{\frac{(-k)_j(k+\lambda+1)_j(\lambda+5/2)_{2j}}{(1)_j(3/2)_j(2\lambda+3)_{2j}}\hypergeom43{-l,l+\lambda+1,j+\frac{\lambda}{2}+\frac{5}{4},j+\frac{\lambda}{2}+\frac{7}{4}}{\frac{3}{2},j+\lambda+2,j+\lambda+\frac{3}{2}}{1}}.
\end{align*}
\normalsize

Note that all of the above hypergeometric functions are balanced. Furthermore for $\lambda=1/2$ one of the terms in the numerator  cancels a denoninator term so they all become balanced ${}_3F_2$'s and can be summed using the Pfaff-Saalschiitz formula. The remaining sums in turn reduce to the Legendre case discussed earlier. 

At this point we are unable to determine whether for certain values of $\lambda$ the above sums simplify or there is any orthogonality as in the Legendre case. Another interesting problem is the asymptotics of the above sums.
\end{remark} 

A formula for the second integral in the above Theorem maybe obtained using equation~(4.7.6) (first formula) in \cite{Szego} and is
\begin{align*}
&\int_{0}^1 \hat p_{k}^{(\lambda)}(2x-1)\hat p_{l}^{(\lambda)}(2x-1)x^{2\lambda-1}\left(\frac{1-x}{1+x}\right)^{\lambda-\frac{1}{2}}dx\\&=(-1)^{k+l}k^{\lambda}_kk^{\lambda}_l\frac{(\lambda+\frac{1}{2})_k(\lambda+\frac{1}{2})_l}{(k+2\lambda)_k(l+2\lambda)_l}\Gamma(\lambda+1/2)\\&\times\sum_{j=0}^i\sum_{n=0}^l\frac{(-k)_j(k+2\lambda)_j(-i)_n(i+2\lambda)_n\Gamma(j+n+2\lambda)}{(1)_j(\lambda+1/2)_j(1)_n(\lambda+1/2)_n\Gamma(j+n+3\lambda+1/2)}\\&\times\hypergeom21{\lambda-1/2,j+n+2\lambda}{j+n+3\lambda+1/2}{-1}.
\end{align*}

The next step is to obtain a generalized eigenvalue problem which will be a 1D-relation for the function $f_{i,j}^{(\lambda)}$ unlike \eqref{dWaveUP}. Our first approach uses the fact that the ultraspherical polynomials satisfy second order differential equations and apparently the approach can be generalized to the case of polynomials satisfying differential equations such as Krall polynomials, Koornwinder's generalized Jacobi polynomials and some Sobolev orthogonal polynomials.

\begin{theorem}\label{UGEPth}
Let $j$ be a fixed nonnegative integer number. Then the function $f=f(i)=f^{(\lambda)}_{i,j}$ of the discrete variable $i$ satisfies the generalized eigenvalue problem 
\begin{multline}\label{recurge}
2((i+\lambda)^2-1/4)(i+\lambda+\frac{3}{2})\sqrt{\frac{i+2\lambda}{(i+1)(i+\lambda+1)(\lambda+i)}}f_{i+1,j}^{(\lambda)}+\\
+2((i+\lambda)^2-1/4)(i+\lambda-\frac{3}{2})\sqrt{\frac{i}{(i-1+\lambda)(i-1+2\lambda)(\lambda+i)}}f_{i-1,j}^{(\lambda)}=\\
(j+\lambda-\frac{1}{2})(j+\lambda+\frac{1}{2})\Big[2(i+\lambda-1/2)\sqrt{\frac{i+2\lambda}{(i+1)(i+\lambda+1)(\lambda+i)}}f_{i+1,j}^{(\lambda)}+4f_{i,j}^{(\lambda)}+\\
2(i+\lambda+1/2)\sqrt{\frac{i}{(i-1+\lambda)(i-1+2\lambda)(\lambda+i)}}f_{i-1,j}^{(\lambda)}\Big],
\end{multline}
for $i=0$, $1$, $2$, \dots and, here, the number $(j+\lambda-\frac{1}{2})(j+\lambda+\frac{1}{2})$ is the corresponding generalized eigenvalue. 
\end{theorem}

\begin{remark} For the case $\lambda=1/2$, formula \eqref{recurge} was obtained in \cite{GM15}. 
\end{remark}

\begin{proof}
To make all the formulas shorter and, more importantly transparent, let us introduce the following operators
\begin{align}\label{tai}
 A_i&=2(i+\lambda+\frac{3}{2})\sqrt{\frac{i+2\lambda}{(i+1)(i+\lambda+1)(\lambda+i)}}E_+\nonumber\\&+2(i+\lambda-\frac{3}{2})\sqrt{\frac{i}{(i-1+\lambda)(i-1+2\lambda)(\lambda+i)}}E_-\nonumber\\&=a_{i+1}E_++a_{i-1}E_-
\end{align}
and
\begin{align}\label{bi}
B_i&=4I+2(i+\lambda-1/2)\sqrt{\frac{i+2\lambda}{(i+1)(i+\lambda+1)(\lambda+i)}}E_+\nonumber\\&+2(i+\lambda+1/2)\sqrt{\frac{i}{(i-1+\lambda)(i-1+2\lambda)(\lambda+i)}} E_-\nonumber\\&=4I+b_{i+1}E_++b_{i-1}E_-,
\end{align}
where $I$ is the identity operator and $E_+$, $E_-$ are the forward and backward shift operators on $i$, respectively. With these notations,  equation~\eqref{recurge} can be rewritten as
\begin{equation}\label{recurge1}
(i(i+2\lambda)+\lambda^2-1/4)A_if_{i,j}^{(\lambda)}=(j(j+2\lambda)+\lambda^2-1/4)B_if_{i,j}^{(\lambda)}
\end{equation}
or
\begin{equation}\label{recurge2}
i(i+2\lambda)A_if_{i,j}^{(\lambda)}+(\lambda^2-1/4)(A_i-B_i)f_{i,j}^{(\lambda)}=j(j+2\lambda)B_if_{i,j}^{(\lambda)}.
\end{equation}
Notice that
\begin{align}\label{aimbi}
A_i-B_i&=-4I+4\sqrt{\frac{i+2\lambda}{(i+1)(i+\lambda+1)(\lambda+i)}}E_+\nonumber\\&-4\sqrt{\frac{i}{(i-1+\lambda)(i-1+2\lambda)(\lambda+i)}}E_-
\end{align}
As is known \cite{Szego}, the ultraspherical polynomials satisfy the differential equation
\begin{equation}\label{diffe}
\frac{d}{dt}((t(1-t))^{\lambda+1/2}\frac{d}{dt}\hat p_j^{(\lambda)}(2t-1))+j(j+2\lambda)(t(1-t))^{\lambda-1/2}\hat p_j^{(\lambda)}(2t-1)=0.
\end{equation}
Thus after two integration by parts we have
%\small
%\begin{align}\label{firstint}
\begin{equation*}%\label{firstint}
    %\resizebox{1\textwidth}{!}{$ 
    \begin{split}
j(j+2\lambda)B_i f_{i,j}^{(\lambda)}=-\int_0^1\frac{d}{dt}((t(1-t))^{\lambda+1/2}\frac{d}{dt}B_i \hat p^{(\lambda)}_i(t)) \hat p^{(\lambda)}_j(2t-1)dt\nonumber\\=-\int_0^1((t(1-t)\frac{d^2}{dt^2}+(\lambda+1/2)(1-2t)\frac{d}{dt})B_i \hat p^{(\lambda)}_i(t))\hat p^{(\lambda)}_j(2t-1)(t(1-t))^{\lambda-1/2}dt\nonumber\\=-\int_0^1((1-t^2)\frac{d^2}{dt^2}-(2\lambda+1)t\frac{d}{dt})B_i \hat p^{(\lambda)}_i(t))\hat p^{\lambda}_j(2t-1)(t(1-t))^{\lambda-1/2}dt\nonumber\\-\int_0^1((t-1)\frac{d^2}{dt^2}+(\lambda+1/2)\frac{d}{dt})B_i \hat p^{(\lambda)}_i(t))\hat p^{\lambda}_j(2t-1)(t(1-t))^{\lambda-1/2}dt.
\end{split}
%$}
    \end{equation*}
%\end{align}
%\normalsize
Now
\[
\begin{split}
-((1-t^2)\frac{d^2}{dt^2}-(2\lambda+1)t\frac{d}{dt})B_i \hat p^{(\lambda)}_i(t))=(i+1)(i+1+2\lambda)b_{i+1} \hat p^{(\lambda)}_{i+1}(t)\\+i(i+2\lambda)b_{i} \hat p^{(\lambda)}_i(t)+(i-1)(i-1+2\lambda)b_{i-1} \hat p^{(\lambda)}_{i-1}(t).
\end{split}
\]
Since
$$
(i\pm 1)(i\pm1+2\lambda)2(i+\lambda\mp \frac{1}{2})-i(i+2\lambda)2(i+\lambda\pm \frac{3}{2})
\mp 4(\lambda^2-1/4)=0,
$$
%and
%$$
%(i-1)(i-1+2\lambda)2(i+\lambda+\frac{1}{2})-i(i+2\lambda)2(i+\lambda-\frac{3}{2})
%+4(\lambda^2-1/4)=0,
%$$
it follows that
\begin{align}\label{recurge4}
&(j(j+2\lambda)B_i-i(i+2\lambda)A_i-(\lambda^2-1/4)(A_i-B_i))f^{(\lambda)}_{i,j}\nonumber\\&=-\int_0^1((t-1)\frac{d^2}{dt^2}+(\lambda+1/2)\frac{d}{dt})B_i \hat p^{(\lambda)}_i(t))\hat p^{(\lambda)}_j(2t-1)(t(1-t))^{\lambda-1/2}dt\nonumber\\&+4(i(i+2\lambda)+\lambda^2-1/4)\int_0^1\hat p^{(\lambda)}_i(t)\hat p^{(\lambda)}_j(2t-1)(t(1-t))^{\lambda-1/2}dt.
\end{align}
We note that
$$
b_{i+1}=4\frac{i+\lambda-1/2}{i+1}a_{i+1}=4(1+\frac{\lambda-3/2}{i+1})a_{i+1}
$$
and
$$
b_{i-1}=4\frac{i+\lambda+1/2}{i+2\lambda-1}a_{i}=4(1-\frac{\lambda-3/2}{i+2\lambda-1})a_{i}.
$$
The substitution of these relations in \eqref{recurge4} leads to the following
\begin{align*}
B_i\hat p^{(\lambda)}_i(t)&=4(1+\frac{\lambda-3/2}{i+1})a_{i+1}\hat p_{i+1}^{(\lambda)}(t)+4(1-\frac{\lambda-3/2}{i+2\lambda-1})a_{i} \hat p_{i-1}^{(\lambda)}(t)+4\hat p_i^{(\lambda)}(t)\\&=
4(1+t+\frac{\lambda-3/2}{i+1}t)\hat p_i^{(\lambda)}(t)-8a_i\frac{(\lambda-3/2)(\lambda+i)}{(i+1)(i+2\lambda-1)}\hat p_{i-1}^{(\lambda)}(t)).
\end{align*}
Using the first equation in \cite[equation~(4.7.28)]{Szego} gives
$$
\frac{d}{dt}\hat p_{i-1}^{(\lambda)}(t)=2\frac{(i+\lambda-1)a_{i}}{i}(t\frac{d}{dt}\hat p_{i}^{(\lambda)}(t)-i \hat p_{i}^{(\lambda)}(t))
$$
so we find
\begin{align*}
\frac{d}{dt}B_i\hat p^{(\lambda)}_i(t)&=
4\frac{d}{dt}(1+t+\frac{\lambda-3/2}{i+1}t)\hat p_i^{(\lambda)}(t)-8a_i\frac{(\lambda-3/2)(\lambda+i)}{(i+1)(i+2\lambda-1)}\frac{d}{dt}\hat p_{i-1}^{(\lambda)}(t))\\&=4(\lambda-1/2)\hat p^{(\lambda)}_i(t)+4(1+t)\frac{d}{dt}\hat p^{(\lambda)}_i(t).
\end{align*}
Thus we have 
\begin{align*}
&((1-t)\frac{d}{dt}-(\lambda+1/2))\frac{d}{dt} B_i\hat p^{(\lambda)}_i(t)\\&=
4((1-t^2)\frac{d^2}{dt^2}-(2\lambda+1)t\frac{d}{dt}-(\lambda^2-1/4))\hat p^{(\lambda)}_i(t)
\end{align*}
and the result follows.
\end{proof}

\begin{remark} At first, we can see that equation \eqref{recurge} has the form
\begin{align*}
\tilde A_if^{(\lambda)}_{i,j}=(j+\lambda-\frac{1}{2})(j+\lambda+\frac{1}{2})B_if^{(\lambda)}_{i,j},
\end{align*}
where
\begin{equation*}
\tilde A_i=(i+\lambda-1/2)(i+\lambda+1/2)A_i,
\end{equation*}
the operators $A_i$ and $B_i$ are given by \eqref{tai} and \eqref{bi}, respectively. At second, the above-given proof shows that the three-term recurrence relation \eqref{recurge} is a consequence of the fact that ultraspherical polynomials are eigenfunctions of a second order differential operator of a specific form. However, there is another way to see the validity of equation \eqref{recurge}.
\end{remark}

We first prove the following statement.

\begin{proposition}\label{ultracoeff} The following representation holds
\begin{equation}\label{Ione}
f^{(\lambda)}_{i,j}=\begin{cases}
0,\,\, i< j;\\
\frac{1}{2^{3j+1}}\sqrt{\frac{i!(\lambda+1)_i(2\lambda)_i(2\lambda)_j}{j!(\lambda)_i(\lambda)_j(\lambda+1)_j}}\frac{(i+2\lambda)_j}{(\lambda+\frac{1}{2})_j(i-j)!}\hypergeom21{-i+j,\ i+j+2\lambda}{2j+2\lambda+1}{\frac{1}{2}},\, i\ge j.
\end{cases}
\end{equation}
\end{proposition}

\begin{proof} Write
\begin{equation}\label{cijlama}
f^{(\lambda)}_{i,j}=k_{i,j,\lambda}\int_0^1 p_i^{(\lambda)}(t)p_j^{(\lambda)}(2t-1)(t(1-t))^{\lambda-1/2}dt,
\end{equation}
where $p_n^{(\lambda)}$ is the monic orthogonal polynomial and
\begin{equation}\label{norm}
k_{i,j,\lambda}=\frac{\Gamma(\lambda+1)}{\Gamma(\lambda+\frac{1}{2})\sqrt{\pi}}2^{i+j+2\lambda+1}\sqrt{\frac{(\lambda)_i(\lambda+1)_i}{i!(2\lambda)_i}}\sqrt{\frac{(\lambda)_j(\lambda+1)_j}{j!(2\lambda)_j}}.
\end{equation}
If we denote the integral in equation~\eqref{cijlama} as $I^{(1)}$ we find using the representation
\[
p^{(\lambda)}_i(t)=2^i\frac{(\lambda+\frac{1}{2})_i}{(i+2\lambda)_i}\hypergeom21{-i,\ i+2\lambda}{\lambda+\frac{1}{2}}{\frac{1-t}{2}},
\]
and set
\begin{equation}\label{ione}
I^{(1)}=2^{i+j}(-1)^j\frac{(\lambda+\frac{1}{2})_i}{((i+2\lambda)_i)}\frac{(\lambda+\frac{1}{2})_j}{((j+2\lambda)_j)}I^{(2)},
\end{equation}
with
\begin{align*}\label{itwo}
I^{(2)}&=\int_0^1\hypergeom21{-i,\ i+2\lambda}{\lambda+\frac{1}{2}}{\frac{1-t}{2}}\hypergeom21{-j,\ j+2\lambda}{\lambda+\frac{1}{2}}{t}(t(1-t))^{\lambda-1/2}dt\\&=\sum_{k=0}^i\frac{(-i)_k(i+2\lambda)_k}{(1)_k(\lambda+\frac{1}{2})_k 2^k}\sum_{n=0}^j\frac{(-j)_n(j+2\lambda)_n}{(1)_n(\lambda+\frac{1}{2})_n}\int_0^1(1-t)^{k+\lambda-1/2}t^{n+\lambda-1/2}dt.
\end{align*}
The integral can be evaluated as $\frac{\Gamma(k+\lambda+\frac{1}{2})\Gamma(n+\lambda+\frac{1}{2})}{\Gamma(k+n+2\lambda+1)}=\frac{(\lambda+\frac{1}{2})_k(\lambda+\frac{1}{2})_n\Gamma(\lambda+\frac{1}{2})^2}{(2\lambda+1)_k(k+2\lambda+1)_n\Gamma(2\lambda+1)}$. From the Chu-Vandermonde formula the sum on $n$ yields
$$
\sum_{n=0}^j\frac{(-j)_n(j+2\lambda)_n}{(1)_n(k+2\lambda+1)_n}=\frac{(k-j+1)_j}{(k+2\lambda+1)_j},
$$
and the sum on $k$ now becomes
$$
\sum_{k=j}^i\frac{(-i)_k(i+2\lambda)_k(k-j+1)_j}{(1)_k(2\lambda+1)_k(k+2\lambda+1)_j2^k}=\sum_{k=0}^{i-j}\frac{(-i)_{k+j}(i+2\lambda)_{k+j}(k+1)_j}{(1)_{k+j}(2\lambda+1)_{k+j}(k+j+2\lambda+1)_j2^{k+j}}.
$$
With the identities
$$
(k+b)_j=\frac{(j+b)_k (b)_j}{(b)_k},\ (a)_{k+j}=(a+j)_k (a)_j,
$$
the above sum becomes
\begin{align*}
&\sum_{k=0}^{i-j}\frac{(-i)_{k+j}(i+2\lambda)_{k+j}(k+1)_j}{(1)_{k+j}(2\lambda+1)_{k+j}(k+j+2\lambda+1)_j2^{k+j}}\\=&\frac{(-i)_j(i+2\lambda)_j(1)_j}{(2\lambda+1)_j(j+2\lambda+1)_j}\sum_{k=0}^{i-j}\frac{(-i+j)_k(i+j+2\lambda)_j}{(1)_j(2j+2\lambda+1)_k}\frac{1}{2^k}\\=&\frac{(-i)_j(i+2\lambda)_j(1)_j}{(2\lambda+1)_j(j+2\lambda+1)_j}\hypergeom21{-i+j,\ i+j+2\lambda}{2j+2\lambda+1}{\frac{1}{2}}.
\end{align*}
Combining all this together gives the result.
\end{proof}

The above hypergeometric representation \eqref{Ione} for $f^{(\lambda)}_{i,j}$ gives a recurrence relation among them. 

\begin{proof}[Another Proof of Theorem \ref{UGEPth}]
To see this use the contiguous relation (see \cite[equation (2.5.15)]{AAR}) 
\begin{align*}
&2b(c-b)(b-a-1)\hypergeom21{a-1,\ b+1}{c}{\frac{1}{2}}\\&-(b-a)(b+a-1)(2c-b-a-1)\hypergeom21{a,\ b}{c}{\frac{1}{2}}\\&-2a(b-c)(b-a+1)\hypergeom21{a+1,\ b-1}{c}{\frac{1}{2}}=0,
\end{align*}
which with $a=-i+j,\ b=i+j+2,$ and $c=2j+2\lambda+1$ yields the relation
\begin{align}%\label{recur}
&(2i+2\lambda-1)(j+2\lambda+1)\sqrt{\frac{i+2\lambda}{(i+1)(i+\lambda+1)(\lambda+i)}}(i+1-j)f^{(\lambda)}_{i+1,j}\nonumber\\&-(2j+2\lambda-1)(2j+2\lambda+1)f^{(\lambda)}_{i,j}\nonumber\\&+(2i+2\lambda+1)(i-j-1)\sqrt{\frac{i}{(i-1+\lambda)(i-1+2\lambda)(\lambda+i)}}(i+j-2\lambda-1)f^{(\lambda)}_{i-1,j}=0.
\end{align}
The latter relation leads to \eqref{recurge}.
\end{proof}

A generalized eigenvalue problem can also be found for $i$ fixed. To this end we need to use the relation 
\[
\hypergeom21{-n,\ b}{c}{x}=\frac{(b)_n}{(c)_n}(-x)^n\hypergeom21{-n,\ -c-n+1}{-b-n+1}{1/x}.
\] 
Therefore we find
\begin{align}\label{althyp}
&\hypergeom21{-i+j,\ i+j+2\lambda}{2j+2\lambda+1}{1/2}\nonumber\\&=\frac{(i+j+2\lambda)_{i-j}}{(2j+2\lambda+1)_{i-j}}(-2)^{j-i}\hypergeom21{-i+j,\ -i-j-2\lambda}{-2i-2\lambda+1}{2}.
\end{align}
Following the steps used to obtain the recurrence formula for $j$ fixed in the second proof we find that 
\begin{align*}
c_j f^{(\lambda)}_{i,j-1}+d_jf^{(\lambda)}_{i,j+1}+e_jf^{(\lambda)}_{i,j}=0,
\end{align*}
where
$$
c_j=-2(i+j+2\lambda-1)(i-j+1)(2j+2\lambda+1)(j+\lambda+1),
$$
$$
d_j=-4(i-j-1)(i+j+2\lambda+1)(j+\lambda-\frac{1}{2})\sqrt{\frac{j(j+1)(j+\lambda-1)(j+\lambda+1)}{(j+2\lambda-1)(j+2\lambda)}},
$$
and
\begin{align*}
e_j&=-2(2i+2\lambda+1)(2i+2\lambda-1)(j+\lambda+1)\sqrt{\frac{j(j+\lambda)(j+\lambda-1)}{(j+2\lambda-1)}}\\&+6(2j+2\lambda-1)(2j+2\lambda+1)(j+\lambda+1)\sqrt{\frac{j(j+\lambda-1)(j+\lambda)}{(j+2\lambda-1)}}.
\end{align*}
Since
$$
(i+j+2\lambda\mp1)(i-j\pm1)=(i+\lambda+\frac{1}{2})(i+\lambda-\frac{1}{2})-(j+\lambda\mp\frac{1}{2})(j+\lambda\mp\frac{3}{2})
$$
%and
%$$
%(i+j+2\lambda+1)(i-j-1)=(i+\lambda+\frac{1}{2})(i+\lambda-\frac{1}{2})-(j+\lambda+\frac{1}{2})(j+\lambda+\frac{3}{2})
%$$
the above recurrence can be recast as the generalized eigenvalue equation
\begin{equation*}
\hat A_jf^{(\lambda)}_{i,j}=(i+\lambda+\frac{1}{2})(i+\lambda-\frac{1}{2})\hat B_jf^{(\lambda)}_{i,j},
\end{equation*}
where the operator $\hat A_j$ is the second order difference operator
\begin{align}\label{hai}
\hat A_j&=(2j+2\lambda+1)(2j+2\lambda-1))\bigg(3(j+\lambda+1)I\nonumber\\&+(j+\lambda+\frac{3}{2})\sqrt{\frac{(j+1)(j+\lambda+1)}{(j+\lambda)(j+2\lambda)}}
\hat E_+\\&+(j+\lambda-\frac{3}{2})(j+\lambda+1)\sqrt{\frac{(j+2\lambda-1)}{j(j+\lambda)(j+\lambda-1)}}\hat E_-\bigg)\nonumber,
\end{align}
the operator $\hat B_j$ is another second order difference operator given by the formula
\begin{align}\label{hbi}
\hat B_j&=4(j+\lambda+1)I\nonumber\\&+(2j+2\lambda-1)\sqrt{\frac{(j+1)(j+\lambda+1)}{(j+\lambda)(j+2\lambda)}}\hat E_+\\&+(2j+2\lambda+1)(j+\lambda+1)\sqrt{\frac{(j+2\lambda-1)}{j(j+\lambda)(j+\lambda-1)}}\hat E_-\bigg)\nonumber,
\end{align}
the operator $I$ is the identity operator, and $\hat E_+$, $\hat E_-$ are the forward and backward shift operators on $j$, respectively. Thus we have just proved the following statement.

\begin{theorem}\label{ifixedTH}
Let $i$ be a fixed nonnegative integer number. Then the function $f=f(j)=f^{(\lambda)}_{i,j}$ of the discrete variable $j$ satisfies the generalized eigenvalue problem 
\[
\hat A_jf^{(\lambda)}_{i,j}=(i+\lambda+\frac{1}{2})(i+\lambda-\frac{1}{2})\hat B_jf^{(\lambda)}_{i,j}
\]
for $i=0$, $1$, $2$, \dots and where the operators $\hat A_j$ and $\hat B_j$ are given by \eqref{hai} and \eqref{hbi}, respectively.
Also, here, $(i+\lambda+\frac{1}{2})(i+\lambda-\frac{1}{2})$ is the corresponding generalized eigenvalue. 
\end{theorem} 
\begin{remark} For the case $\lambda=1/2$, Theorem \ref{ifixedTH} was obtained in \cite{GM15}. 
\end{remark}

Using the asymptotic results for the Gauss hypergeometric function from \cite{J01} and \cite{W18} (see also \cite{P1}, \cite{T03}) one can easily get asymptotic behavior of the solution $f^{(\lambda)}_{i,j}$ for  $j$ fixed and when $i$ tends to infinity.

\begin{theorem}
For sufficiently large $i$ the following formula holds
\begin{equation}\label{asymi}
f^{(\lambda)}_{i,j}=k_j\frac{\cos\left(\pi\left(j + \frac{\lambda}{2} - \frac{i}{2} + \frac{1}{4}\right)\right)}{\sqrt{\pi}i^{\lambda + 1/2}}+O\left(\frac{1}{i^{\lambda+3/2}}\right),
\end{equation}
where
\begin{equation}\label{ksubj}
k_j=\frac{1}{2^{j+1-2\lambda}}\sqrt{\frac{(2\lambda)_j}{j!(\lambda)_j(\lambda+1)_j\lambda\Gamma(2\lambda)}}\Gamma(2j+2\lambda+1)(\lambda+\frac{1}{2})_j .
\end{equation}
\end{theorem}
\begin{proof}
According to Proposition \ref{ultracoeff} for $i\ge j$ we have
\begin{equation*}
f^{(\lambda)}_{i,j}=\frac{1}{2^{3j+1}}\sqrt{\frac{i!(\lambda+1)_i(2\lambda)_i(2\lambda)_j}{j!(\lambda)_i(\lambda)_j(\lambda+1)_j}}\frac{(i+2\lambda)_j}{(\lambda+\frac{1}{2})_j(i-j)!}\hypergeom21{-i+j,\ i+j+2\lambda}{2j+2\lambda+1}{\frac{1}{2}}.
\end{equation*}
Then, since 
$$
\sqrt{\frac{i!(\lambda+1)_i(2\lambda)_i}{(\lambda)_i}}\frac{(i+2\lambda)_j}{(i-j)!}=\sqrt{\frac{1}{\lambda \Gamma(2\lambda)}} i^{(2j+\lambda)}(1+O(1/i)), 
$$
formula \eqref{asymi} follows from \cite[formula (36)]{J01}.
\end{proof}
\begin{remark} Formula \eqref{asymi} along with the fact that $f^{(\lambda)}_{i,j}=0$ for  $i<j$ show that the moving wave behavior of the solution demonstrated in  Figure \ref{Fig1} is also characteristic for the solution $f^{(\lambda)}_{i,j}$ of the discrete wave equation \eqref{dWaveUP} for any $\lambda>-1/2$. 
\end{remark}

Another useful asymptotic is when $i=k_1 t$ and $j=k_2 t$ where $k_1>k_2$ are fixed and $t$ is large. 
\begin{theorem}
For $k_1 t$ and $k_2 t$  integers with  $k_1>k_2>0$, and $\frac{\sqrt{2}k_2}{k_1}>1$
\begin{equation}\label{fasyf}
f^{(\lambda)}_{k_1 t,k_2 t}=\frac{c(\epsilon,\lambda)}{2^{k_1 t+1} (k_1 t)^{\frac{1}{2}}}\left(\frac{1+\hat b(\epsilon)}{\epsilon-\hat b(\epsilon)}\right)^{(k_1-k_2)t}\left(\frac{1+2\epsilon-\hat b(\epsilon)}{1+\epsilon}\right)^{(k_1+k_2)t+2\lambda}(1+O(1/t)),
\end{equation}
where
\begin{equation}\label{ck1k2}
c(\epsilon,\lambda)=\epsilon^{\lambda}\frac{1}{\sqrt{\pi(1-\epsilon^2)(2\epsilon^2-1)^{\frac{1}{2}}}},
\end{equation}
$\epsilon=\frac{k_2}{k_1}$, and $\hat b(\epsilon)=\sqrt{2\epsilon^2-1}$. 
\end{theorem}
\begin{proof}
In this case the representation given by equation~\eqref{althyp} is most convenient. An application of the transformation T3 in \cite{P1} yields
\begin{align*}
&\hypergeom21{-i+j,\ -i-j-2\lambda}{-2i-2\lambda+1}{2}\nonumber\\&=\frac{2^{i-j-1}(i-j)!}{(i+j+2\lambda+1)_{i-j-1}}\hypergeom21{-i+j+1,\ i+j+2\lambda+1}{2}{1/2}\nonumber\\&=-\frac{2^{i+j+2\lambda-1}(i-j)!}{(i+j+2\lambda+1)_{i-j-1}}\hypergeom21{i-j+1,\ -i-j-2\lambda+1}{2}{1/2},
\end{align*}
where Euler's transformation has been used to obtain the last equality. 
Thus with the use of the duplication formula for the $\Gamma $ function it follows
\begin{equation}\label{asyf}
f^{(\lambda)}_{i,j}=d_{i,j}\hypergeom21{i-j+1,\ -i-j-2\lambda+1}{2}{1/2},
\end{equation}
where
$$
d_{i,j}=(-1)^{i-j+1}2^{j+2\lambda-1}\sqrt{\frac{(i+\lambda)(j+\lambda)i!\Gamma(2\lambda+j)}{j!\Gamma(2\lambda+i)}}.
$$
This becomes
\begin{align}\label{dijtl}
d_{k_1 t,k_2 t}&=(-1)^{(i-j+1}2^{j+2\lambda-1}\left(\frac{j^{2\lambda}}{i^{2\lambda-2}}\right)^{1/2}(1+O(1/i)\nonumber\\&=(-1)^{(k_1-k_2)t+1}2^{k_2 t+2\lambda-1}\left(\frac{k_2}{k_1}\right)^{\lambda} (k_1 t) (1+O(1/t)).
\end{align}

The hypergeometric function on the right hand side of equation~\eqref{asyf} is in the form to use the type B formulas in \cite{P1} and leads to considering the hypergeometric function $\hypergeom21{\epsilon_1 w+1,\ -w-2\lambda+1}{2}{1/2}$ where $\epsilon_1 w$ is an integer. Equation~(4.4) in \cite{P1} shows that the saddle points occur at $\frac{1+\epsilon_1}{2}\pm\sqrt{(\frac{1+\epsilon_1}{2})^2-2\epsilon_1}$. If the discriminant is positive both saddles are real and equation (4.9) in \cite{P1} yields
\begin{align*} 
&\hypergeom21{\epsilon_1 w+1,\ -w-2\lambda+1}{2}{1/2}\\&=\frac{(-1)^{\epsilon_1 w+1}}{w^{\frac{3}{2}}\sqrt{\pi\epsilon_1 b(\epsilon_1)}}(\frac{1+\epsilon_1+b(\epsilon_1)}{1-\epsilon_1-b(\epsilon_1)})^{\epsilon_1 w}\frac{(3-\epsilon_1-b(\epsilon_1))^{w+2\lambda}}{2^{w+4\lambda+\frac{1}{2}}}(1+O(1/w)),
\end{align*}
where
\begin{equation}\label{bep1}
b(\epsilon_1)=\sqrt{(1+\epsilon_1)^2-8\epsilon_1}.
\end{equation}
With $\epsilon_1=\frac{k_1-k_2}{k_1+k_2}$ and $w=(k_1+k_2)t$ the above equations yield \eqref{fasyf}.

\end{proof}

\begin{remark}
When  the discriminant is negative, the  two saddle points are conjugates of each other and so in this case equation~(4.7) in \cite{P1} is used to obtain the asymptotics for $\hypergeom21{\epsilon_1 t+1,\ -t-2\lambda+1}{2}{1/2}$ which then are used to obtain the asymptotics of $f^{(\lambda)}_{k_1 t,k_2 t}$.
\end{remark}

We finish this section with a couple of statements where we start with the recurrence formulas. Write the recurrence formula in equation~\eqref{recurge} as
\begin{equation}\label{iirecur}
a_{i,j}f^{(\lambda)}_{i+1,j}+b_{i,j}f^{(\lambda)}_{i,j}+c_{i,j}f^{(\lambda)}_{i-1,j}=0,
\end{equation}
and the recurrence formula in $j$ as
\begin{equation}\label{jjrecur}
\hat a_{i,j}f^{(\lambda)}_{i,j+1}+\hat b_{i,j}f^{(\lambda)}_{i,j}+\hat c_{i,j}f^{(\lambda)}_{i,j-1}=0,
\end{equation} with $i\ge j\ge0$.

We can now prove  the following simple statement.
\begin{proposition}\label{iiirecur}
Given $a_{i,j}$, $b_{i,j}$, $c_{i,j}$ and $\lambda>-1/2$. For each $j>0$ the unique solution of equation~\eqref{iirecur} with initial conditions
\[
f_{j-1,j}=0, \quad f_{j,j}=\int_0^1 \hat p^{(\lambda)}_j(t)\hat p^{(\lambda)}_j(2t-1)(t(1-t))^{\lambda-1/2}dt
\]
is the function
\[
f_{i,j}=I^{(\lambda)}_{i,j}:=\int_0^1 \hat p^{(\lambda)}_i(t)\hat p^{(\lambda)}_j(2t-1)(t(1-t))^{\lambda-1/2}dt.
\]
 If $j=0$, 
 $\lambda>-1/2$, and $\lambda\ne1/2$ then $f_{0,0}=I_{0,0}^{(\lambda)} $ gives the unique solution $f_{i,0}=I^{(\lambda)}_{i,0}$. If $\lambda=1/2$ then the initial conditions $f_{0,0}= I^{(1/2)}_{0,0}$ and $f_{1,0}= I^{(1/2)}_{1,0}$ are needed to give  $f_{i,j}=I^{(1/2)}_{i,j}$.
\end{proposition}
\begin{proof}
For $j>0,\ a_{i,j}\ne 0$ for $i\ge j$ so the result follows from equation~\eqref{iirecur}. For $j=0$ and $\lambda\ne1/2,\ c_{0,0}=0\ne a_{0,0}$ so that only $f_{0,0}$ is needed to compute $f_{1,0}$. The remaining $f_{i,j}$ are computed in the standard fashion from equation~\eqref{iirecur}. For the last case when $\lambda=1/2,\ a_{0,0}=0=b_{i,0}$ so $f_{2,0}=\frac{c_{1,0}}{a_{1,0}} f_{0,0}$ and $f_{3,0}=\frac{c_{2,0}}{a_{2,0}} f_{1,0}$. The remaining $f_{i,0}$ are computed in the same way using the fact that $a_{i,0}\ne0$ for $i>0$.
\end{proof}
Similarly, for the recurrence in $j$ we have the following.
\begin{proposition}\label{jjjrecur}
Given $a_{i,j}$, $b_{i,j}$, $c_{i,j}$ and $\lambda>-1/2$, for each $i>0$ the unique solution of equation~\eqref{jjrecur} with initial conditions
$f_{j,j+1}=0$ and $f_{j,j}=I_{j,j}^{(\lambda)}$ is $f_{i,j}=I^{(\lambda)}_{i,j}$. 
\end{proposition}
Since $\hat c_{i,j},\ \hat b_{i,j}$, and $\hat a_{i,j} $ are not equal to zero for $i\ge j$ the result follows from equation~\eqref{jjrecur}.

\section{Connections to other problems}

Recall that it is said that a function $\Psi(x,y)$ is a solution of a bispectral problem if it satisfies the following
\[
\begin{split}
A\Psi(x,y)&=g(y)\Psi(x,y)\\
B\Psi(x,y)&=f(x)\Psi(x,y),
\end{split}
\]
where $A$, $B$ are some operators, with $A$ acting only on $x$ and $B$ acting only on $y$, and $f$, $g$ are some functions \cite{DG}. It is shown in \cite{LM98} that if $A$ and $B$ are tridiagonal operators then the solutions of the corresponding discrete bispectral problem are related to the Askey-Wilson polynomials. 

The problem we are dealing with in this paper is the following generalization of a bispectral problem:
\begin{equation}\label{gBP}
\begin{split}
A\Psi(i,j)&=g(j)B\Psi(i,j)\\
C\Psi(i,j)&=f(i)D\Psi(i,j),
\end{split}
\end{equation}
where $i$, $j$ are discrete variables, the operators $A$ and $B$ are tridiagonal operators acting on the index $i$, and  $C$, $D$ are tridiagonal operators acting on the index $j$. Note that each equation in \eqref{gBP} is a generalized eigenvalue problem and, hence, the problem \eqref{gBP} includes a bispectral problem as a particular case (for instance when $B$ and $D$ are the identity operators). 

Setting $\Psi(i,j)=f^{(\lambda)}_{i,j}$ we see that Theorems \ref{UGEPth} and \ref{ifixedTH} tell us that $f^{(\lambda)}_{i,j}$ is a solution of a generalized bispectral problem of the form \eqref{gBP}. Actually, it would be nice to find a characterization of such generalized bispectral problems similar to what was done in \cite{LM98} for discrete bispectral problems. It would also be interesting to study the consistency relations for the system \eqref{jjjrecur}, \eqref{iiirecur} and those relations will constitute a nonlinear system of difference equations on the coefficients of  \eqref{jjjrecur}, \eqref{iiirecur}.

Another link that is worth discussing here is the relation to linear spectral transformations. To see this in its simplest form, let us consider two families $\hat{p}_j^{(1/2)}(t)$ and $\hat{p}_j^{(3/2)}(t)$ of the ultraspherical polynomials. That is, we consider the two measures on $[-1,1]$
\[
d\mu_{1/2}(t)=dt, \quad d\mu_{3/2}(t)=(1-t^2)dt,
\]
which are clearly related in the following manner
\begin{equation}\label{GeronimusNM}
d\mu_{1/2}(t)=\frac{d\mu_{3/2}(t)}{(1-t^2)}.
\end{equation}
In such a case, one usually says that $d\mu_{1/2}(t)$ is a Geronimus transformation of $d\mu_{3/2}(t)$ of the second order or $d\mu_{1/2}(t)$ is the inverse quadratic spectral transform of $d\mu_{3/2}(t)$ (for instance, see \cite{MandCo11}). As a matter of fact, the Geronimus transformation of $d\mu_{3/2}(t)$ is more general than just \eqref{GeronimusNM} and it has the form
\[
d\mu^{(G)}(t)=\frac{d\mu_{3/2}(t)}{(1-t^2)}+M\delta_{-1}+N\delta_{1},
\]
where $\delta_{a}$ denotes the Dirac delta function supported at $a$ and $M$, $N$ are some nonegative real numbers. For the corresponding orthogonal polynomials we have that
\[
\hat{p}_i^{(G)}(t)=\alpha(1,i)\hat{p}_i^{(3/2)}(t)+\alpha(2,i)\hat{p}_{i-1}^{(3/2)}(t)+\alpha(3,i)\hat{p}_{i-2}^{(3/2)}(t),
\]
where $\alpha(1,i)$, $\alpha(2,i)$, and $\alpha(3,i)$ are some coefficients and they are of the form \eqref{Cijg}. For instance, in the simplest case \eqref{GeronimusNM}, introducing the coefficients 
\[
f_{i,j}^{(1/2,3/2)}=\int_{-1}^1\hat p_i^{(1/2)}(t)\hat p_j^{(3/2)}(t)(1-t^2)dt
\]
leads to the relation
\begin{equation}\label{Geronimus}
\hat{p}_i^{(1/2)}(t)=f_{i,i}^{(1/2,3/2)}\hat{p}_i^{(3/2)}(t)+f_{i,i-1}^{(1/2,3/2)}\hat{p}_{i-1}^{(3/2)}(t)+f_{i,i-2}^{(1/2,3/2)}\hat{p}_{i-2}^{(3/2)}(t),
\end{equation}
where
\[
f_{i,i-2}^{(1/2,3/2)}=-\sqrt{\frac{k_{i-2,i-2,3/2}}{k_{i,i,1/2}}},\quad 
f_{i,i-1}^{(1/2,3/2)}=0 \quad
f_{i,i}^{(1/2,3/2)}=\sqrt{\frac{k_{i,i,1/2}}{k_{i,i,3/2}}}
\]
and $k_{i,j,\lambda}$ is defined by formula \eqref{norm}. In fact, this can be generalized to the case of arbitrary Geronimus transformation but the formulas will get messier.

Due to Theorem \ref{GdWaveTH}, the coefficients $f_{i,j}^{(1/2,3/2)}$ satisfy the discrete wave equation in question. Besides, formula \eqref{Geronimus} shows that in the sequence $f_{i,j}^{(1/2,3/2)}$ when $j$ is fixed there are at most three nonzero coefficients and we know how to find them explicitly. Moreover, returning to the moving wave interpretations we did before we see that in this case we have a localized wave and below is the simulation.
\begin{figure}[h!]
\includegraphics[width=\linewidth]{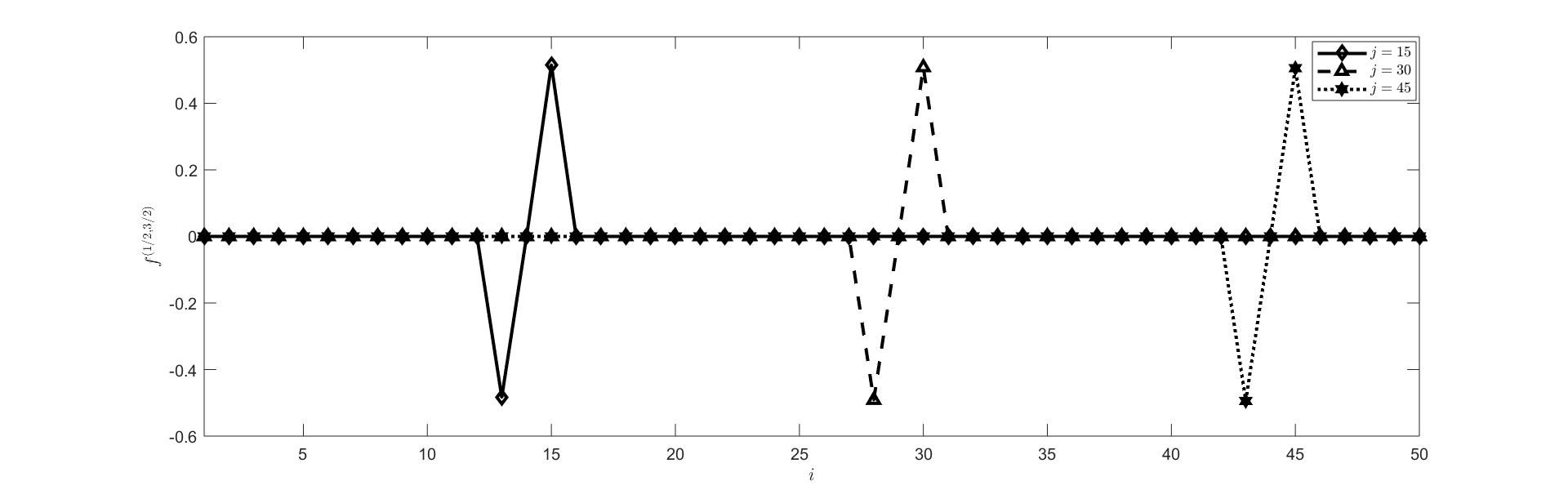}
\caption{This picture shows three graphs of the function $f=f^{(1/2,3/2)}(i)=f^{(1/2,3/2)}_{i,j}$ of the discrete space variable $i$ at the three different discrete times $j=15$, $j=30$, and $j=45$.}
\label{Fig3}
\end{figure}

The phenomenon of localized waves is related to the fact that the measures are related to one another through spectral transformations. Still, one can define even more general coefficients 
\[
f_{i,j}^{(\lambda,\mu)}=\int_{-1}^1\hat p_i^{(\lambda)}(t)\hat p_j^{(\mu)}(t)(1-t^2)^{\mu-1/2}dt
\]
and, as before, they form a solution to a wave equation. Moreover, these coefficients are known explicitly \cite[Section 7.1]{AAR} and are called the connection coefficients. It will be shown in a forthcoming paper that the family $f_{i,j}^{(\lambda,\mu)}$ is also a solution of a bispectral problem of the form \eqref{gBP}. In addition, the coefficients
\[
f_{i,j,k}^{(1,1,1)}=\frac{2}{\pi}\int_{-1}^1 p_i^{(1)}(t) p_j^{(1)}(t)p_k^{(1)}(t)(1-t^2)^{1/2}dt,
\]
where the polynomials $p_i^{(1)}(t)$ are the monic Chebyshev polynomials of second kind, count Dyck paths \cite{deSMV85}. Thus, it would be interesting to find out if the coefficients $f_{i,j,k}^{(1,1,1)}$ still satisfy some generalized eigenvalue problems and in which case if such generalized eigenvalue problems admit a combinatorial interpretation.

\medskip

\noindent{\bf Acknowledgments.} M.D. was supported in part by the NSF DMS grant 2008844. The authors are grateful to Erik Koelink for interesting and helpful remarks. They are also indebted to the anonymous referees for suggestions that helped to improve the presentation of the results. J.G. would like to thank J.G. for the support.

\end{document}